\documentclass[%
 reprint,
superscriptaddress,
 amsmath,amssymb,
 aps,
]{revtex4-2}

\usepackage{graphicx}
\usepackage{dcolumn}
\usepackage{bm}

\usepackage{hyperref}
\usepackage{amsfonts, amsmath, amsthm, amssymb} 
\usepackage{braket}
\usepackage{float}
\usepackage{amsthm}
\usepackage{csquotes}
\usepackage{ragged2e} 

\newtheorem{lemma}{Lemma}

\usepackage{xcolor}

\newcommand{\tr}[1]{\text{tr}\left(#1\right)}

\usepackage{bbold}
\usepackage{overpic}
\begin{document}

\title{Certifying ergotropy under partial information}
\author{Egle Pagliaro}
\email{egle.pagliaro@icfo.eu}
\affiliation{ICFO - Institut de Ciencies Fotoniques, The Barcelona Institute of Science and Technology, 08860 Castelldefels, Barcelona, Spain}
\author{Leonardo Zambrano}
\affiliation{ICFO - Institut de Ciencies Fotoniques, The Barcelona Institute of Science and Technology, 08860 Castelldefels, Barcelona, Spain}
\author{Mir Alimuddin}
\affiliation{ICFO - Institut de Ciencies Fotoniques, The Barcelona Institute of Science and Technology, 08860 Castelldefels, Barcelona, Spain}
\author{Alioscia Hamma}
\affiliation{Physics Department E. Pancini, Università degli Studi di Napoli Federico II, Complesso Universitario Monte S. Angelo, Via Cintia, 80126 Napoli, Italy
}
\affiliation{Istituto Nazionale di Fisica Nucleare (INFN), Sezione di Napoli, Complesso Universitario Monte Sant’Angelo, Via Cintia, 80126 Napoli, Italy
}
\author{Antonio Ac\'in}
\affiliation{ICFO - Institut de Ciencies Fotoniques, The Barcelona Institute of Science and Technology, 08860 Castelldefels, Barcelona, Spain}
\affiliation{ICREA, Passeig Lluis Companys 23, 08010 Barcelona, Spain}
\author{Donato Farina}
\email{donato.farina@unina.it}
\affiliation{Physics Department E. Pancini, Università degli Studi di Napoli Federico II, Complesso Universitario Monte S. Angelo, Via Cintia, 80126 Napoli, Italy
}
\affiliation{Istituto Nazionale di Fisica Nucleare (INFN), Sezione di Napoli, Complesso Universitario Monte Sant’Angelo, Via Cintia, 80126 Napoli, Italy
}

\date{\today}

\begin{abstract}
Ergotropy, the maximum work extractable from a quantum system, is a central resource in quantum physics. Computing ergotropy is well established when the system state is fully known, but its estimation under partial information remains an open problem. Here we introduce a general certification framework that lower bounds ergotropy using only the expectation values of a limited set of arbitrary observables. The method naturally applies in the finite-statistics regime, yielding confidence-certified bounds that explicitly incorporate shot noise. We benchmark our approach on both synthetic data and experimental measurements from an IBM quantum processor. This establishes a robust and experimentally accessible tool for certifying extractable work in realistic quantum settings.

\end{abstract}

\maketitle

{\it Introduction.---}%
Quantum thermodynamics seeks to understand how quantum systems store, transform, and exchange energy~\cite{binder2018thermodynamics,vinjanampathy2016quantum,goold2016role,landi2021irreversible,deffner2019quantum,francica2022role}. 
Beyond its foundational interest, this field underpins emerging technologies such as quantum batteries~\cite{alicki2013quantum,binder2015quantacell,campaioli2017colloquium,shi2022entanglement,le2018spin,friis2018precision,quach2022superabsorption,Yang2023} and quantum thermal machines~\cite{kosloff2014quantum,uzdin2015equivalence,hofer2015quantum,pekola2015towards,brandner2020thermodynamic}.
A central question in this context is how much work can be extracted from a quantum system, thereby identifying work-resourceful states and distinguishing them from passive, resourceless ones \cite{RevModPhys.91.025001, gour2024resources}.
This resource is typically quantified by the \emph{ergotropy}~\cite{allahverdyan2004maximal}, the maximum energy extractable by unitary operations and representing the ultimate ``charge'' that a quantum battery can deliver in an ideal, entropy-preserving process.

While ergotropy is analytically computable when both the quantum state and the Hamiltonian are perfectly known in their spectral decompositions~\cite{allahverdyan2004maximal}, this assumption is hardly realistic in practical settings. 
Experimentally, any attempt to characterize a many-body system is limited by measurement constraints, noise, and finite sampling. Full state tomography quickly becomes infeasible~\cite{haah2017sample,guta2020fast}, meaning that the state from which work is to be extracted is generally only {partially characterized}. 
In such informationally incomplete settings, even determining whether any work can be extracted is a nontrivial task. 
Recent works have explored ergotropy under incomplete information, including notions such as observational and Boltzmann ergotropy~\cite{vsafranek2023work}, and analyses of work extraction in random or partially known states~\cite{chakraborty2025sample, hovhannisyan2024concentration, canzio2025extracting}. 
However, a general and {certified} method for bounding ergotropy directly from partial measurement data, avoiding full tomography, remains lacking. 
This limitation hinders our ability to realistically quantify and benchmark quantum energy extraction protocols in experiments. 

\begin{figure}
\centering
\includegraphics[width=0.75\linewidth]{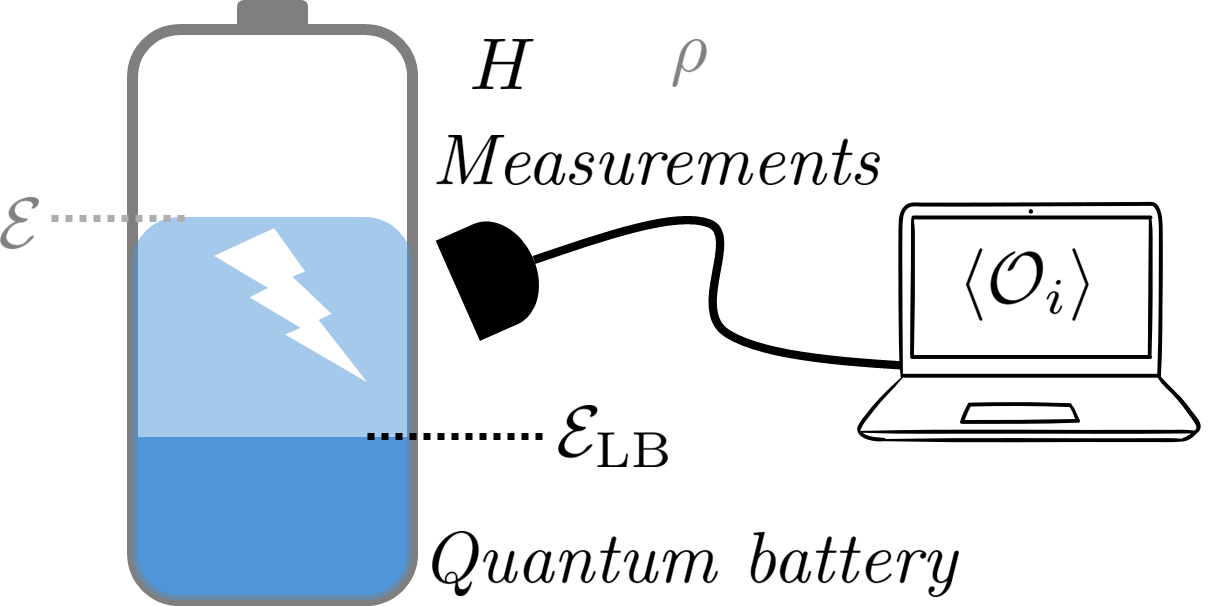}
\caption{ 
Schematic illustrating the problem under consideration. We assume the Hamiltonian $H$ of the system is fully characterized and we want to find a lower bound $\mathcal{E}_{\rm LB}$ for the unknown ergotropy $\mathcal{E}$ under partial information on the \textit{true} state $\rho$. This is obtained through an informationally incomplete set of measurements, {defined by a set of expectation values $\{ \langle \mathcal{O}_i \rangle = {\rm tr}(\rho \mathcal{O}_i)\}_{i\,\in \mathcal{I}}$}.}
\label{fig:cartoon}
\end{figure}

In this work we fill this gap by introducing a framework to {certify ergotropy under partial information}.
Our method provides a rigorous lower bound to this resource compatible with an arbitrary set of measured observables, requiring no assumptions beyond the knowledge of the Hamiltonian (see Fig.\,\ref{fig:cartoon}).
The bound is computed through a {two-step semidefinite programming protocol}, consistent with the available data.
Semidefinite programs (SDPs) are peculiar convex optimization problems that minimize linear objective functions~\cite{vandenberghe1996semidefinite, boyd2004convex}. They converge to the global optimum, a property that makes them particularly suitable for certification tasks in quantum information {theory}~\cite{skrzypczyk2023semidefinite}.
They are efficiently solvable, with computation time scaling polynomially in the number of variables and constraints.
Our approach leverages these aspects and naturally extends to realistic finite statistics scenarios leading to bounds that
hold with a prescribed confidence level.
Its efficacy is demonstrated on both synthetic and real data, yielding an experimentally accessible tool to certify the presence of extractable work in quantum settings.

{\it Framework.---}%
We start by briefly introducing the concepts and notation required to formalize these results.
For a finite-dimensional system in a state $\rho$ with Hamiltonian $H$, ergotropy is defined as the maximum energy that can be extracted via unitary operations,
\begin{equation}
\label{def:ergotropy}
    \mathcal{E}(\rho, H) := \max_{U \in \mathcal{U}(d)} \left[ \tr{\rho H} - \tr{H U \rho U^\dag} \right],
\end{equation}
with $\mathcal{U}(d)$ denoting the set of unitaries on a Hilbert space of dimension $d$.

For closed systems undergoing a unitary process, the average work is identified, through the first law, with the corresponding average energy change. Thus, for a fixed unitary $U$, the quantity ${\rm tr}[\rho(H-U^\dagger H U)]$ represents the average work extracted by that process. This should not be interpreted as a general work observable, nor as defining a measurement of a full work probability distribution\,\cite{PhysRevE.75.050102, RevModPhys.83.771}; throughout this work, we focus exclusively on average work, whose optimization over unitary processes gives the ergotropy\,\cite{allahverdyan2004maximal}.
When both $\rho$ and $H$ are known, the optimization in
Eq.\,\eqref{def:ergotropy} admits a closed-form solution~\cite{allahverdyan2004maximal}. 
Let $\rho= \sum_{j} r_{j} \ket{r_{j}}\bra{r_{j}} $ be the spectral decomposition of $\rho$, with $r_{j+1}\leq r_{j}$, and $H=\sum_{j} E_{j} \ket{E_j}\bra{E_j} $ the one of the Hamiltonian, with $E_{j+1} \geq E_{j}$. The optimal unitary is 
\begin{equation}
    U_{\star} = \sum_j \ket{E_j}\bra{r_j},
    \label{optimalU}
\end{equation}
which maps $\rho$ to its associated \emph{passive state},
$\rho_p=U_{\star}\rho U_{\star}^\dag=\sum_j r_j|E_j\rangle\!\langle E_j|${, diagonal in the energy eigenbasis and with eigenvalues in decreasing order associated
to increasing energies.} This yields a compact expression for the ergotropy in terms of the spectral decompositions, {see Eq.~\eqref{def:ergotropy}},
\begin{equation}
\mathcal{E}(\rho, H) =
\sum_{i,j} E_i r_j \,|\braket{E_i|r_j}|^2 
-
\sum_i r_i E_i.
\end{equation}

{{\it Setup.---}%
In this work we assume that the Hamiltonian $H$ is fully known. However, the state $\rho$ is only partially known, having available only a finite set of expectation values of observables 
$\{\mathcal{O}_i\}_{i\in\mathcal{I}}$,
\begin{equation}
\label{constraints-meas}
    o_i = \tr{\rho \, \mathcal{O}_i}, \qquad i \in \mathcal{I}.
\end{equation}
Here, $\mathcal{I}$ is an index set with cardinality $K := |\mathcal{I}|$, the number of measured observables.
The information \eqref{constraints-meas}, together with state positivity and the trace-one condition, define the \emph{feasible set} of compatible states through the appropriate constraints,
\begin{equation}\label{eq:feasible_set}
\Omega_{\mathcal{I}} = \{ X : X \succeq0,\ {\rm tr}(X) = 1,\ {\rm tr}(X \mathcal{O}_i) = o_i\ \forall i \in \mathcal{I} \}.
\end{equation}
This set is convex and guaranteed to contain the true state $\rho$. Our goal is to certify a nontrivial lower bound of ergotropy over the above feasible set.}

{While certifying ergotropy from incomplete data is generally nontrivial, partial analytical progress is possible in restricted settings (see Appendix C and \cite{suppmat}). For instance, Canzio \textit{et al.} recently derived explicit lower bounds when the {only} measured observable is the Hamiltonian~\cite{canzio2025extracting}.
A key simplifying feature behind such analytic results is that the available data fix the mean energy, i.e., $\tr{HX}$ is constant over $X\in\Omega_{\mathcal I}$: in that case, minimizing ergotropy reduces to maximizing the passive energy, effectively a classical problem.  
Notably, if one can additionally access the energy eigenbasis measurement, the bound may even be tightened to the incoherent ergotropy (see Appendix C) through our approach.
However, outside these settings the feasible set typically spans different energies, so the minimization must balance $\tr{HX}$ against the passive contribution; with finite statistics, this is further complicated because the constraints defining $\Omega_{\mathcal I}$ hold only within confidence regions.
This additionally motivates addressing the problem in full generality, as we show below.}
\begin{figure}
\centering
\includegraphics[width=.9\linewidth]{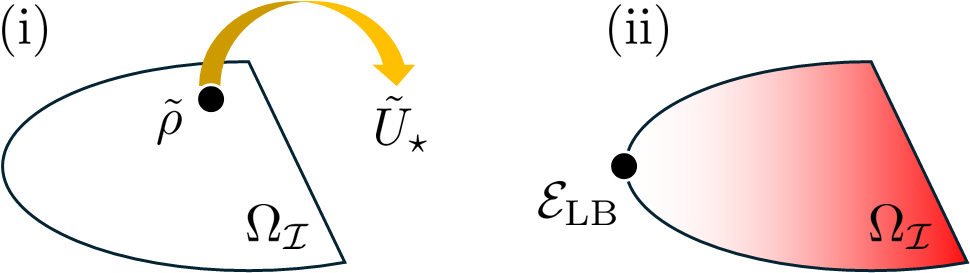}
\caption{ 
Schematic representation of the two-step procedure for determining the ergotropy lower bound. 
{In step (i), we select a state compatible with the measurement constraints and construct the corresponding optimal unitary for work extraction. Keeping this unitary fixed, step (ii) determines the worst-case scenario by minimizing the energy difference over the state. This procedure yields the desired lower bound on the ergotropy.}}
\label{fig:scheme}
\end{figure}

{\it Certification of ergotropy: two-step protocol.---}%
The smallest ergotropy compatible with the available information can be expressed in the following minimax form:
\begin{align}\label{eq:bound_ergotropy_minimax}
\bar{\mathcal{E}}_{\rm LB}=
\min_{X \in \Omega_{\mathcal{I}}} \max_{U \in \mathcal{U}(d)} \big[\tr{H X} - \tr{H U X U^\dag}\big]\,,
\end{align}
Although the problem is non-convex due to the maximization over unitaries, it can be relaxed into a tractable form.
We notice that, by construction, any unitary \( U \) provides a lower bound on the quantity\,\eqref{eq:bound_ergotropy_minimax}, and therefore for the ergotropy,
namely, for any given $U$,
\begin{align}
\min_{X \in \Omega_{\mathcal{I}}} & \left[\tr{H X} - \tr{H U X U^\dag}\right]
\leq
\bar{\mathcal{E}}_{\rm LB}
\leq \mathcal{E}\,.
\label{trivial-bound}
\end{align}
Starting from this observation, we design our practical protocol schematically depicted in Fig.\,\ref{fig:scheme}. 
It consists of two steps: (i)\,we select a suitable unitary; then, (ii)\,we compute the corresponding lower bound by minimizing over all quantum states in \(\Omega_{\mathcal{I}}\) while keeping the unitary fixed. 
Specifically, the protocol is structured as follows:
\begin{enumerate}
    \item[(i)] We solve the SDP
    \begin{eqnarray}
    \label{sdp}
    &&\tilde{\rho} = {\rm argmin}\,\, \ell(X)\\
    &&\hspace{.7cm} X \in \Omega_{\mathcal{I}}, \nonumber
    \end{eqnarray}
    where \(\ell\) is any linear function (or linearizable function such as state purity\,\cite{suppmat}). 
    Essentially,
    this step selects a state within the feasible set \(\Omega_{\mathcal{I}}\). Then, from the spectral decomposition $\tilde{\rho} = \sum_i \tilde{r}_i \ket{\tilde{r}_i}\bra{\tilde{r}_i}$, 
    we compute the unitary,
    \begin{equation}
    \label{Utildestar}
    \tilde{U}_{\star} = \sum_j \ket{E_j}\bra{\tilde{r}_j}, 
    \end{equation}
    optimal in terms of ergotropy for $\tilde\rho$, cf. Eq.~\eqref{optimalU}.
    
    \item[(ii)] Fixing \(\tilde{U}_{\star}\) from the previous step, we solve the SDP
    \begin{eqnarray}
    \label{ergoLB}
    &&\mathcal{E}_{\rm LB} = \min_{X \in \Omega_{\mathcal{I}}} \left[\tr{H X} - \tr{H \tilde{U}_{\star} X \tilde{U}_{\star}^\dag}\right].
    \end{eqnarray}
    Since \(\tilde{U}_{\star}\) is unitary and because of Eq.~\eqref{trivial-bound}, this quantity provides a certified lower bound for the ergotropy,
    \begin{equation}
    \label{eq:bound}
    \mathcal{E}_{\rm LB} \leq \mathcal{E}.
    \end{equation}
\end{enumerate}
Several important properties follow directly from the construction of the bound
$\mathcal{E}_{\rm LB}$ in Eq.~\eqref{ergoLB}.
First, \emph{convergence} is guaranteed as the amount of available information increases.
When the measurement data become informationally complete, the feasible set
$\Omega_{\mathcal{I}}$ collapses to a single element, namely the true state $\rho$,
and the protocol recovers the true ergotropy,
$\mathcal{E}_{\rm LB} = \mathcal{E}(\rho,H)$.
Second, in the opposite regime of severe informational incompleteness, the feasible set
may be so large that the optimization in Eq.~\eqref{ergoLB} yields a negative value.
Since ergotropy is by definition non-negative, we then assign the \emph{trivial bound}
$\mathcal{E}_{\rm LB}=0$, corresponding to the  {impossibility of certifying the presence of} extractable work.
Third, although one would intuitively expect the certified bound to improve monotonically
as more information is supplied, this behavior is not automatically guaranteed by a literal
application of steps (i) and (ii), since the unitary is recomputed at each stage.
However, \textit{monotonicity} can be enforced
by augmenting the protocol with a simple conditional criterion: the unitary is updated in step (i) only if it yields a strictly tighter bound than the previous iteration; otherwise, the previous unitary is retained (see more details in Appendix A). 

{\it Finite statistics.---}%
Finally, the framework can be adapted to include finite statistics effects. 
In realistic experiments, the expectation values of the observables \(\mathcal{O}_i\) are not known exactly but are estimated from a finite number of measurement shots.
Each measurement shot requires an independent copy of the system, so that the total number of shots quantifies the sample complexity of the problem.
Let \(o_i^{\rm (est)}\) denote the empirical average obtained from \(N_i\) measurement outcomes.  
To account for statistical fluctuations, we adapt the framework of Ref.~\cite{zambrano2024certification} to the problem of ergotropy certification, and modify the feasible set so that it contains, with high probability, all states compatible with the measurement data.  
Specifically, the feasible set \(\Omega_{\mathcal{I}}\) becomes 
\begin{eqnarray}
\label{feasible-set-conc-ineq}
&&\Omega_{\mathcal{I}} :=
\{ X : X \succeq 0,\, {\rm tr}(X)=1, \\
&&\hspace{1cm} \vert \tr{X \mathcal{O}_i} - o_i^{\rm (est)} \vert \leq \epsilon_i \,\, \forall i \in \mathcal{I}\}\,.
\nonumber
\end{eqnarray}
Assuming, for simplicity, {observables with} eigenvalues in the interval \([-1,1]\)  (as is the case for the Pauli strings considered below), the deviation is given by
\begin{equation}
\epsilon_i := \sqrt{\frac{2 \log(2K/\delta)}{N_i}}.
\label{eq:epsilon}
\end{equation}
Here, $K=|\mathcal{I|}$ as before, and \(1 - \delta\) is the desired confidence level.  
As a consequence, the bound obtained from Eq.~\eqref{eq:bound} holds now with a probability of at least \(1 - \delta\).  
This provides a rigorous way to incorporate finite-statistics effects in our problem. A detailed derivation of Eq.~\eqref{eq:epsilon} based on Hoeffding’s inequality and the union bound is given in Appendix~B.

\begin{figure}
    \centering
    \includegraphics[width=1\linewidth]{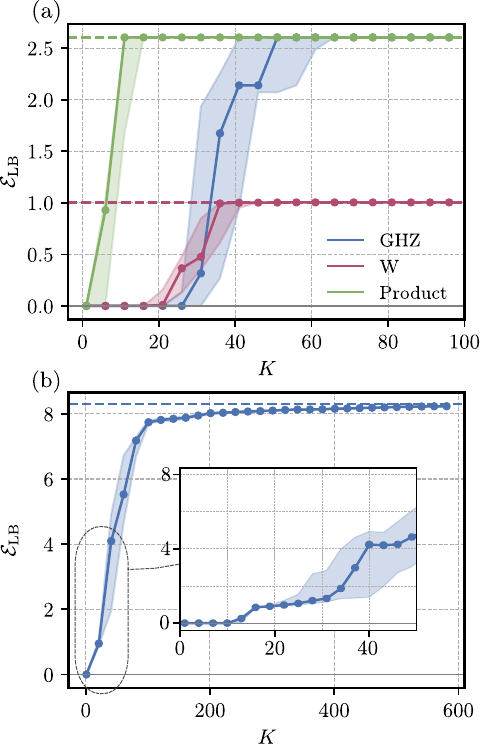}
    \caption{(a)\ Certified ergotropy lower bounds 
as a function of the number  $K$ of measured observables (random Pauli strings).
We considered the true states
\eqref{pure-states}, GHZ (blue), product (green), and W (red), and XXZ Hamiltonian ($J_1=1$, $J_y=1$, $\Delta=0.5$, $B=0$).
(b)\,Same as in (a) but for the  negative-temperature state \eqref{negative-gibbs} ($\beta=-1$) and MFI Hamiltonian ($B=0.5$, $G=0.5$, $\Delta=1$). 
For each value of $K$, each realization is generated according to the hierarchical sampling described in the text.
The plotted curves show the median over 20 realizations, with shaded regions indicating the interquartile range.
Dashed horizontal lines mark the true ergotropies. We considered a system of 5 qubits and numerically simulated measurement data in the absence of shot noise}.
    \label{fig:no_stat}
\end{figure}

{\it Examples.---}%
The following use cases illustrate how our approach yields useful lower bounds on the ergotropy across a wide range of informationally incomplete scenarios.
This includes considering several combinations of Hamiltonians and true states to illustrate its general applicability.

Let us consider an $n$-qubit system. The multiqubit Pauli strings $\{P_i\}$, consisting of all the $4^n$ tensor products of $\{\mathbb{1},\sigma^x,\sigma^y,\sigma^z\}$, form an orthogonal operator basis under the Hilbert–Schmidt inner product. Any density matrix or observable can therefore be expanded as a linear combination of Pauli strings.
The Hamiltonian can likewise be written as
$
    H = \sum_{i} h_i P_i ,
$
with real coefficients $h_i$ fixed by the physical model.
More specifically, we focus on the widely used family of spin-chain Hamiltonians of the form
\begin{eqnarray}
\nonumber
&&
H=-J_1 \sum_{i=1}^{n-1} \sigma^x_i\sigma^x_{i+1}
-J_2\sum_{i=1}^{n-2}\sigma^x_i\sigma^x_{i+2}
-B\sum_{i=1}^{n} \sigma_i^z\\
&&
-G \sum_{i=1}^{n} \sigma_i^x
-J_y \sum_{i=1}^{n-1} \sigma^y_i\sigma^y_{i+1}
-\Delta \sum_{i=1}^{n-1} \sigma^z_i\sigma^z_{i+1}\,,
\label{spin-local-ham}
\end{eqnarray}
having assumed open boundary conditions.
This encompasses paradigmatic many-body models, including the Axial Next-Nearest Neighbor Ising (ANNNI) model ($G=J_y=\Delta=0$), the XXZ Heisenberg model with transverse field ($G=J_2=0, J_y=J_1$), and the Mixed Field Ising (MFI) model ($J_1=J_2=J_y=0$). 

In each proof-of-principle example, we fix a reference state $\rho$ taken as the true state.
This state is only partially known via the input information (see Eq.\,\eqref{constraints-meas}).
The measurements that we choose are Pauli measurements, namely
$o_i = \tr{\rho \, {P}_i}$, $ i \in \mathcal{I}.$
We consider a structured random order for Pauli strings: first all one-body operators (in random order), then all two-body operators (in random order), and so on.
We refer to this random sampling strategy as \textit{hierarchical}.
With this choice, considering $K$ Pauli measurements corresponds to selecting the first $K$ operators from this body-ordered random list.
The lower bounds $\mathcal{E}_{\rm LB}$ that we {report}   (specifically, in Fig.\,\ref{fig:no_stat} and Fig.\,\ref{fig:real_stat}(a))
correspond to medians over many realizations, with shaded bands indicating the interquartile ranges.
Furthermore, in the step (i) of our protocol we minimize state purity ${\rm tr}(X^2)$, a choice that typically provides tighter bounds (see Appendix C and \cite{suppmat} for details).

\begin{figure}
    \centering
    \includegraphics[width=1\linewidth]{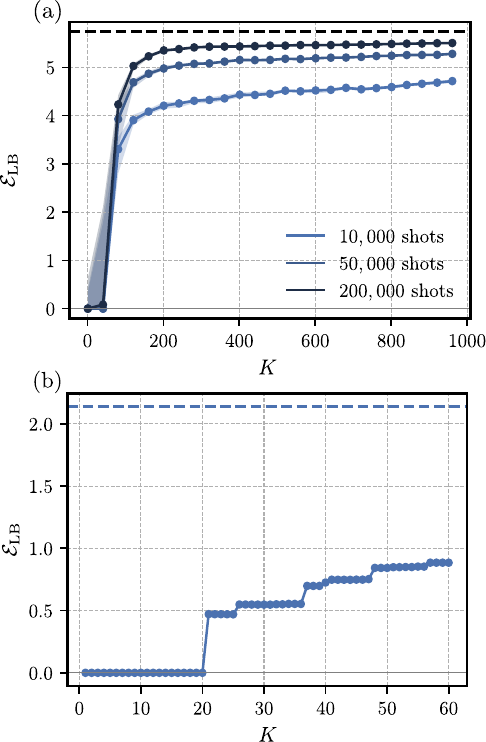}
    \caption{Probabilistic ergotropy lower bounds as a function of the number $K$ of measured observables (confidence level $1-\delta=0.997$).
(a) Effect of finite sampling on the  ergotropy lower bound, considering a  five-qubit ANNNI model ($J_1=1$, $J_2=-1$, $B=0.5$) and the true state \eqref{superposition}, superposition of extremal Hamiltonian eigenstates. 
Curves correspond to different numbers of simulated measurement shots per Pauli observable (see legend). 
For each $K$, the point shown is the median over 20 realizations using hierarchical sampling.
Shaded regions denote the interquartile range across realizations.
Dashed horizontal lines indicate the true ergotropies.
(b)\,Ergotropy lower bounds obtained from experimental data for a four-qubit GHZ state and considering XXZ Hamiltonian ($J_1=J_y=1$, $\Delta=0.5$, $B=0$). Here the Pauli constraints are incorporated in the fixed experimental order, without resampling; each expectation value is estimated from $2^{14}$ measurement shots. Each point corresponds to a single deterministic realization, and the optimal unitary is updated only when it yields a strictly tighter bound, ensuring monotonicity of the certified curve.
Here, the reference ergotropy (dashed line) is computed on the target GHZ state.}
\label{fig:real_stat}
\end{figure}

In
Fig.\,\ref{fig:no_stat}(a) we report  ergotropy lower bounds $\mathcal{E}_{\rm LB}$
as a function of $K$, the number of Pauli measurements considered
(here we consider numerically simulated data in the absence of shot noise).
This is done for three distinct pure true states---GHZ, W, product state---namely $\rho=\ket{\psi}\bra{\psi}$,
with $\psi \in \{\text{GHZ, W, prod\}}$
and where
\begin{eqnarray}
\label{pure-states}
&& \ket{\text{GHZ}}
=
\frac{1}{\sqrt{2}}
\big(
\ket{0}^{\otimes n}
+
\ket{1}^{\otimes n}
\big)\,,\\
&& 
\ket{\text{W}}
=
\frac{1}{\sqrt{n}}
\big(
\ket{10\dots 0}
+
\ket{010\dots 0}
+
\cdots
+
\ket{0\dots 01}
\big),
\nonumber
\\
&&
\ket{ {\mathrm{prod}}}
=
\bigotimes_{i=1}^{n} \ket{0}    
\,.
\nonumber
\end{eqnarray}
We consider the XXZ Hamiltonian
and a system of $n=5$ qubits.
The bounds tighten monotonically as more Pauli constraints are included, abruptly converging toward the true ergotropy values.
This behavior is reminiscent of quantum compressed sensing theory \cite{PhysRevLett.105.150401}, whose domain are low rank states. 
Fig.\,\ref{fig:no_stat}(b) reports the same but considering as true state a negative-temperature Gibbs state,
\begin{equation}
\label{negative-gibbs}
\rho = \frac{e^{-\beta H}}{{\rm tr}(e^{-\beta H})}\,, \quad \beta < 0\,,
\end{equation}
and the MFI Hamiltonian.  
{Here the convergence to the
true value with the number of observables is smoother}, as a consequence of the state’s full-rank property.  
However, also in this case very few measurements are sufficient to retrieve nonzero values for $\mathcal{E}_{\rm LB}$, i.e. to tag the state as resourceful (see inset of Fig.\,\ref{fig:no_stat}(b)).

Fig.~\ref{fig:real_stat} shows how the certification protocol performs when applied to measurement data in the presence of shot noise. Specifically, Fig.~\ref{fig:real_stat}(a) considers synthetic (numerically simulated) data, while Fig.~\ref{fig:real_stat}(b) uses experimental data.
Here we have probabilistic ergotropy lower bounds that hold true under a prescribed confidence level.
In Fig.~\ref{fig:real_stat}(a), 
we consider the ANNNI Hamiltonian and a true state prepared as a coherent superposition of extremal Hamiltonian eigenstates, namely $\rho=\ket{\psi}\bra{\psi}$ and 
\begin{equation}
\label{superposition}
\ket{\psi}=\frac{1}{\sqrt{2}}
(\ket{\epsilon_{1}}+\ket{\epsilon_{d}})\,,
\end{equation}
under finite measurement statistics. 
Simulated data are generated using a binomial distribution \cite{suppmat}.
Shot noise is incorporated through the confidence-interval constraints introduced earlier in Eq.\,\eqref{eq:epsilon}, which typically enlarge the feasible set and therefore loosen the certified bounds. As the number of shots increases, statistical uncertainty shrinks and the certified bounds concentrate around the ideal-statistics curve. Notably, even under strong shot noise (e.g., $10^4$ shots), the protocol reliably returns non-trivial lower bounds, demonstrating robustness to simulated measurement data.
Fig.~\ref{fig:real_stat}(b) presents instead results for experimental data corresponding to a four-qubit GHZ target state, and considering XXZ Hamiltonian. Data are obtained by performing experiments using the \texttt{ibm\_perth} quantum processor. 
Here the measured Pauli strings are imposed in the fixed order determined by the experiment, without random reshuffling\,\cite{suppmat}.

This feature, combined with the fact that the experimentally prepared GHZ state retains an approximately small Pauli support, explains the marked step-like behavior of our lower bound in $K$.
Importantly, the largest value considered is $K=60$ Pauli strings, corresponding to a relatively small fraction ($\approx23.5\%$) of the total (excluding the identity) $255$ possible Pauli strings for a four-qubit system.
Moreover, to ensure the monotonicity of $\mathcal{E}_{\rm LB}$, the optimal unitary obtained at step (i) is updated only when it yields a strictly improved lower bound, otherwise the previously certified unitary is reused (see Appendix\,A).
Despite the presence of experimental imperfections and finite statistics ($2^{14}$ shots per Pauli observable), the bound tightens steadily as more constraints are included, ultimately certifying a substantial fraction of the true ergotropy. 
This demonstrates that the protocol remains effective in a fully realistic setting, where measurement noise and device errors coexist.

{\it Discussion.---}%
In summary, we have introduced a general protocol producing lower bounds on ergotropy under partial information.
The framework remains valid in the finite-statistics regime, where the bounds naturally incorporate shot noise and hold with a prescribed confidence level. 
This provides an experimentally relevant method for certifying the presence of extractable work in realistic quantum batteries.

Our work opens several directions for future research.
In particular, we have assumed perfect knowledge of the system Hamiltonian; a natural extension is to relax this assumption and combine the present protocol with Hamiltonian-learning techniques. 
Our certification of ergotropy is also directly relevant to platforms that engineer spin-chain models, including Rydberg-atom arrays and trapped-ion setups.
Furthermore, the approach could be extended to certify local ergotropies \cite{PhysRevA.107.012405, PhysRevLett.133.150402, PhysRevA.111.012212} leveraging scalable semidefinite-programming relaxations~\cite{Navascues2007NPA, navascues2008njp, PhysRevX.14.031006, hbrt-cn8q}. 

\vspace{.2cm}

{\it Acknowledgments.---}%
This work was supported by the Government of Spain (Severo Ochoa CEX2019-000910-S, FUNQIP, Quantum in Spain and European Union NextGenerationEU PRTR-C17.I1), Fundació Cellex, Fundació Mir-Puig, Generalitat de Catalunya (CERCA program), the ERC AdG CERQUTE, the EU (PASQuanS2.1 101113690 and Quantera Veriqtas projects), the AXA Chair in Quantum Information Science, 
and the PNRR MUR Project No. PE0000023-NQSTI.
E.P. acknowledges support from a “la
Caixa” Foundation (ID 100010434, code LCF/BQ/DI23/11990078). M.A. acknowledges funding from the European Union (QURES, 101153001). 
DF acknowledges financial support from University of Catania via PNRR-MUR Starting Grant project PE0000023-NQSTI.
We acknowledge the use of IBM Quantum services for this work. The views expressed are those of the authors, and do not reflect the official policy or position of IBM or the IBM Quantum team.

\vspace{.2cm}
{\it Data availability.---}%
The data that support the findings of this article are openly available at \cite{pagliaro_gitlab_2026}.

\vspace{.3cm}
\clearpage

{\Large\textbf{End Matter}}

\vspace{.3cm}
\noindent

\noindent
\textbf{Appendix A: Monotonicity.}
\label{app:monotonicity}
Despite intuitively adding
more information should provide a better bound $\mathcal{E}_{\rm LB}$, this is not guaranteed under a literal implementation of
steps (i) and (ii) of our protocol. However, such monotonicity can be restored using the following conditional criterion.

Let us consider two measurement constraint sets, identified by the index sets $\mathcal{I}$ and $\mathcal{I}'$, with
$\mathcal{I}\subseteq\mathcal{I}'$. Adding constraints reduces the set of compatible
states, namely $\Omega_{\mathcal{I}'}\subseteq\Omega_{\mathcal{I}}$.
When computing the bound \eqref{ergoLB} 
associated to $\Omega_{\mathcal{I}'}$, we can first evaluate it using the
unitary $\tilde{U}_\star$ that we have used when dealing with the set
$\Omega_{\mathcal{I}}$. We obtain necessarily a value $\mathcal{E}_{\rm LB}'(\mathcal{I'})\geq \mathcal{E}_{\rm LB}(\mathcal{I})$, hence the desired monotonic behavior.
This follows because the objective function in Eq.\,\eqref{ergoLB} is unchanged, but the minimization is performed over a subset of the previous feasible set.
In essence, this corresponds to skipping directly to step (ii) with the prior unitary.
We can then compute $\mathcal{E}_{\rm LB}(\mathcal{I'})$ following both steps (i), extracting another unitary $\tilde{U}_\star'$, and (ii), both based on the set 
$\Omega_{\mathcal{I}'}$.
Now the objective function in Eq.\,\eqref{ergoLB} can change if $\tilde{U}_\star'\neq \tilde{U}_\star$.
Accordingly,
if and only if   
$\mathcal{E}_{\rm LB}'(\mathcal{I'}) > \mathcal{E}_{\rm LB}(\mathcal{I'})$
we 
consider 
$\mathcal{E}_{\rm LB}'(\mathcal{I'})$ as our updated bound,
instead of the usual $\mathcal{E}_{\rm LB}(\mathcal{I'})$,
and we keep track of the fact that the optimal unitary remains the 
unitary we had already used for the set
$\Omega_{\mathcal{I}}$.
Iterating the routine provides the desired hierarchy of lower bounds.
\vspace{.3cm}

\noindent
\textbf{Appendix B: Confidence regions.}
\label{sec:appendixB}Following 
Refs.\,\cite{PhysRevLett.124.100401, zambrano2024certification}, we derive Eq.~\eqref{eq:epsilon} from Hoeffding’s inequality and the union bound, and show how our bounds obtained from Eq.~\eqref{ergoLB} are certified to hold with probability at least $1-\delta$.

We consider $K$ measured observables $\mathcal{O}_i$, each with eigenvalues in the interval $[-1, 1]$. Let $o_i = \tr{\rho \, \mathcal{O}_i}$ be the true expectation value and let $o_i^{\rm (est)}$ denote the empirical average obtained from $N_i$ measurement shots.
For a single observable $\mathcal{O}_i$, Hoeffding's inequality bounds the probability that $o_i^{\rm (est)}$ deviates from $o_i$ by more than $\epsilon_i$:
\begin{equation}
\Pr\big(|o_i^{\rm (est)} - o_i| \geq \epsilon_i\big) \le 2\exp\left(-{N_i\epsilon_i^2}/{2}\right).
\end{equation}

To ensure that all $K$ estimates simultaneously satisfy their bounds, we apply the union bound. We require the total failure probability to be no more than $\delta$:
\begin{equation}
\Pr\left(\bigcup_{i=1}^K \{|o_i^{\rm (est)} - o_i| > \epsilon_i\}\right) \le \sum_{i=1}^K \Pr\big(|o_i^{\rm (est)} - o_i| > \epsilon_i\big)  \le \delta.
\end{equation}
This is done by enforcing that each term is bounded as $2\exp(-{N_i\epsilon_i^2}/{2}) = {\delta}/{K}$. Solving this equation for $\epsilon_i$ yields the bound for the error of each empirical mean. Then, with probability at least $1-\delta$, all observables satisfy $|o_i^{\rm (est)} - o_i| \leq \epsilon_i$, where
$\epsilon_i$  
is given by Eq.\,\eqref{eq:epsilon}.

Incorporating these constraints in the feasible set $\Omega_{\mathcal{I}}$, Eq.\,\eqref{ergoLB}, guarantees that the true state $\rho$ is contained within this set, $\rho \in \Omega_{\mathcal{I}}$, with a probability of at least $1-\delta$. Consequently, since $\mathcal{E}_{\rm LB}$ is computed by minimizing over this certified set, the resulting bound $\mathcal{E}_{\rm LB} \leq \mathcal{E}(\rho, H)$ holds with a probability of at least $1-\delta$.
\vspace{.3cm}
\noindent

\begin{figure}[H]
\centering
\includegraphics[width=0.9\linewidth]{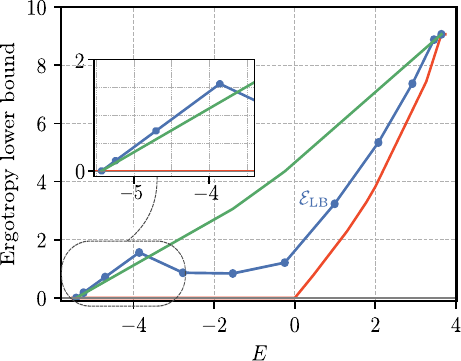}
\caption{Ergotropy lower bound as a function of the mean energy using our two-step protocol (blue curve). 
The considered observables are the Pauli strings appearing in the Pauli decomposition of the Hamiltonian (associated to nonzero components).
We also show the Hamiltonian-only ergotropy lower bound of Ref.~\cite{canzio2025extracting} (red)
and
the energy-basis measurement bound of Eq~\eqref{theorem:IC} (green). 
We considered a four-qubit system, XXZ Hamiltonian ($J_1=J_y=-1$, $\Delta=-0.5$, $B=0$) and numerically simulated measurement data in the absence of shot noise.
The true state is parametrized as a superposition of the lowest- and highest-energy eigenstates, $|\psi(s)\rangle \propto |\epsilon_1\rangle + s|\epsilon_d\rangle$.
Notably, our approach certifies the presence of coherent ergotropy at low energies, as highlighted in the inset, where the blue curve lies above the green one. }
    \label{fig:canzio}
\end{figure}

\noindent
\textbf{Appendix C: Measuring the mean energy.}\label{sec:appendixC}
In the feasible set $\Omega_\mathcal{I}$ defined in Eq.~\eqref{eq:feasible_set}, when the mean energy $\tr{HX}$ is fixed, certifying ergotropy reduces to finding the state $X\in\Omega_\mathcal{I}$ whose corresponding passive state $X_p$ has \emph{maximum} energy:
\begin{equation*}
\bar{\mathcal{E}}_{\rm LB}
\;=\;
[\tr{HX} - \max_{X \in \Omega_{\mathcal{I}}} \tr{H\,X_p}].
\end{equation*}
Note that the passive-state energy is a Schur-concave function~\cite{Alimuddin2020}. 
In particular, the majorization relation~\cite{Marshall} implies
\begin{equation}\label{majorization}
    X \succ Y \implies \tr{HX_p} \leq \tr{HY_p}.
\end{equation}
Therefore, one eventually looks for the \emph{most disordered} state in $\Omega_{\mathcal{I}}$ (equivalently, the state with the highest passive-state energy), since it certifies the minimum ergotropy over the feasible set under the fixed-mean-energy constraint.

In a generic scenario, however, the most disordered state in a feasible set need not be unique: there may exist multiple states whose passive states are \emph{incomparable} under majorization (i.e., neither majorizes the other). In such cases, analytically maximizing the passive-state energy---and hence certifying $\mathcal{E}_{\mathrm{LB}}$---can be challenging. Recently, Canzio et al.~\cite{canzio2025extracting} addressed this issue for feasible sets obtained from a {single} observable, the system Hamiltonian, and derived an analytic lower bound on the minimum ergotropy despite this non-uniqueness. It is also important to note that, for a quantum state $X$, ergotropy can be decomposed into incoherent (population) and coherent contributions~\cite{Francica'2020}:
\begin{equation}
    \mathcal{E}(X,H)=\mathcal{E}_{IC}(X)+\mathcal{E}_C(X),
\end{equation}
where $\mathcal{E}_C(X)=0$ when $X$ has no coherence in the energy eigenbasis. When only partial information is available through the Hamiltonian expectation value, the construction in~\cite{canzio2025extracting} yields a lower bound on the ergotropy that cannot exceed the incoherent ergotropy threshold, i.e. it is   $\le \mathcal{E}_{IC}(X)$ for the true (unknown) state $X$. In many cases, however, this bound is not tight. For instance, if the mean energy satisfies $\tr{HX}\le \tr{H}/d$, then there always exists a thermal state $\tau_\beta\in\Omega_\mathcal{I}$. This follows because, at fixed mean energy, the maximum-entropy (most disordered) state has the Gibbs form, i.e., it is a thermal state with $\beta=1/(k_BT)\ge 0$, and it can realize energies up to $\tr{H}/d$ in this regime. Consequently, the above fact immediately implies $\bar{\mathcal{E}}_{\mathrm{LB}}=0$.

In what follows, we explain how a fine-grained measurement in the energy eigenbasis mitigates this issue and leads to a stronger lower bound (see Fig.~\ref{fig:canzio}). Let $\{\mathcal{O}_i=\ket{\epsilon_i}\!\bra{\epsilon_i}~~s.t. ~~ \sum_i \mathcal{O}_i = \mathbb{1}_d\}$ be the energy-basis measurement. The partial information obtained from an unknown state $X$ is the probability vector $\vec{p}=\{p_i\}$ with
\[
p_i \;=\; \tr{\ket{\epsilon_i }\!\bra{\epsilon_i }\,X},\qquad i=1,\dots,d~.
\]
With $H=\sum_{i=1}^d \epsilon_i \ket{\epsilon_i}\!\bra{\epsilon_i}$, every state $\omega\in\Omega_{\mathcal I}$ (i.e., compatible with $\vec{p}$) has the same mean energy:
\[
\tr{\omega H}\;=\tr{\rho^* H}\;=\; \sum_{i=1}^d p_i\,\epsilon_i \;=\; \tr{HX},
\]
where $\rho^* = \sum_i p_i \ket{\epsilon_i}\!\bra{\epsilon_i}$ is the dephased (energy-diagonal) state and belongs to $\Omega_{\mathcal I}$. In the next Lemma \ref{lemma1}, we show that for a fixed energy-basis distribution $\vec{p}$, the state $\rho^*$ is majorized by every other compatible state in $\Omega_{\mathcal I}$ and hence saturates the ergotropy lower bound. Moreover, the resulting value coincides exactly with the incoherent ergotropy $\mathcal{E}_{IC}(X)$. Finally, note that a thermal state need not reproduce the same probability distribution $\vec{p}$; consequently, it may lie outside the feasible region $\Omega_{\mathcal I}$. This is precisely why incorporating the full energy-basis measurement statistics can yield a strictly better bound than the Hamiltonian-only certification proposed in~\cite{canzio2025extracting}. 

\begin{lemma}\label{lemma1}
Given energy-basis probabilities $\vec p=(p_1,\dots,p_d)$ and the feasible set $\Omega_{\mathcal I}$, the ergotropy is minimized by the dephased state $\rho^\ast\in\Omega_{\mathcal I}$ and the tight lower bound is
\begin{equation}\label{theorem:IC}
\bar{\mathcal{E}}_{\mathrm{LB}}=\mathcal{E}(\rho^\ast,H)
=
\sum_{i=1}^d p_i\,\epsilon_i
-
\sum_{i=1}^d p_i^{\downarrow}\,\epsilon_i
=
\mathcal{E}_{IC}(X),
\end{equation}
where $\vec p^{\,\downarrow}$ is the nonincreasing rearrangement of $\vec p$.
\end{lemma}
The statement is rigorously proven in the  Supplemental Material \cite{suppmat}.
From both the above cases, we understand the following necessary condition for the observable set $\{\mathcal{O}_i\}_{i\in\mathcal I}$ to capture a nonzero \emph{coherent} ergotropy: one must choose observables that do not commute with the Hamiltonian and its energy basis. Moreover, if $\{\mathcal{O}_i\}_{i\in\mathcal I}$ contains mutually non-commuting observables (so $[\mathcal{O}_i,\mathcal{O}_j]\neq 0$ for some $i\neq j$), we access more information about the state and $\mathcal{E}_{LB}$ can increase further as the feasible set $\Omega_{\mathcal I}$ shrinks. It is possible to give a closed-form expression for $\bar{\mathcal{E}}_{LB}$ which is strictly greater than $\mathcal{E}_{IC}(X)$ for a canonical qubit system: there, a nonzero coherent ergotropy necessarily contributes, arising from two mutually non-commuting observables. Here again, the uniqueness of the most disordered state simplifies the problem \cite{suppmat}.

However, in a generic case—higher dimensions or many-body systems—the analytical solution is not always tractable. In general, the feasible set may fail to contain a single state that is majorized by all others which complicates the task of identifying the minimum-ergotropy state and of deriving nontrivial analytic bounds. In such settings, we adopt our generic {two-step protocol [Eqs.\,\eqref{sdp} and \eqref{ergoLB}]. 
In the first step (\ref{sdp}), we determine the state $\tilde{\rho}$ (and the associated unitary $\tilde{U}_{\star}$) by
minimizing state purity (or, equivalently, maximizing the linear entropy) over the feasible set $\Omega_{\mathcal I}$. 
Although $\tilde{\rho}$ may not have the highest passive-state energy, the maximization of linear entropy provides a unique solution $\tilde{\rho}$ and makes it sufficiently disordered (typically yielding higher-rank states), which imposes nontrivial restrictions on the corresponding unitary $\tilde{U}_{\star}$. 
In the second step (\ref{ergoLB}), we fix this unitary and obtain the lower bound by solving the semidefinite program over $\Omega_{\mathcal I}$.

Although our formalism is quite generic and does not require a fixed mean-energy constraint, in the following we illustrate how it works under the fixed-energy constraint. 
We consider a four-qubit system with an XXZ Hamiltonian and choose the Pauli strings appearing in the Pauli decomposition of the Hamiltonian
$
H=\sum_{i \in \mathcal{I}} h_i P_i 
$
as our observables. 
Notice that the sum is over the set $\mathcal{I}$, that in this case also identifies the nonzero Pauli components of $H$.
For Hamiltonians as in \eqref{spin-local-ham}, this implies $K\sim O(n)$.
This choice immediately enforces the energy constraint; however, the observables need not commute with each other, nor with the Hamiltonian. As a demonstration, we choose a probe state $|\psi(t)\rangle \propto |\epsilon_1\rangle + t|\epsilon_d\rangle$, parametrically spanning energies between the ground and the highest excited levels. In Fig.~\ref{fig:canzio}, we explore how this two-step protocol
compares with the Hamiltonian-based protocol~\cite{canzio2025extracting} and the energy-basis measurement protocol (Eq.~\eqref{theorem:IC}) for certifying the minimum ergotropy. Note that our protocol is not tight in general; nevertheless, it yields a better ergotropy lower bound than the Hamiltonian-based method and, at low energies, can certify the presence of coherent ergotropy.
}

\bibliography{sample}

\clearpage
\onecolumngrid
\appendix

\begin{center}
\textbf{\large Supplemental Material}
\end{center}
\vspace{0.3cm}

\setcounter{equation}{0}
\renewcommand{\theequation}{S\arabic{equation}}

\section*{S1. Further details on step\,(i)}
\label{supp_mat:min_purity}

Semidefinite programs are quite general in the sense that many problems not originally posed as SDPs admit equivalent SDP formulations.
One such example is the minimum-purity problem.
The minimum purity over the feasible set $\Omega_{\mathcal{I}}$, 
\begin{align}
\label{min-purity-no-sdp}
\min_{X} \;& {\rm tr}(X^2) \\
&\text{s.t. } X \in \Omega_{\mathcal{I}}\,,
\nonumber
\end{align}
can be cast as the semidefinite program,
\begin{align}
\min_{X,\, Y} \;& {\rm tr}(Y) \\
\text{s.t. } \;& 
\begin{pmatrix}
Y & X \\
X & \mathbb{1}
\end{pmatrix} \succeq 0 \, , \quad
X \in \Omega_{\mathcal{I}} .
\nonumber
\end{align}
Taking an optimal solution $(X^\#,Y^\#)$ of this program, we can use $X^\#$ as the feasible $\tilde{\rho}$ for our step (i).
However, we remark that \textit{any} linear function $\ell$ is formally well suited for step\,(i), as its goal is 
to pick a feasible matrix, namely a quantum state compatible with the measurement constraints.
A particularly natural choice is the zero function, which maps every state to $0$.
However, different choices of the objective can lead, after the second step, to different bounds. Despite all such bounds are valid, we would clearly prefer it to be the tightest possible, i.e., more concretely, our goal is to maximize the ergotropy lower bound for a given set of measurement constraints. 

Considering the purity as the objective function to minimize typically may provide (slightly) tighter bounds. 
In this regard, we notice that the minimum purity problem can be \textit{linearized} as we have just seen, ensuring the problem remains a SDP. 
Also, the minimum purity problem\,\eqref{min-purity-no-sdp} admits a unique \enquote{argmin} solution as, geometrically, it is the Euclidean projection of the zero matrix over a convex set (see Fig.\,\ref{fig:proj}). 
To see this more transparently,
we rewrite the problem as 
$ 
{\rm argmin}_{X}  \Vert X-0\Vert_F^2 
$ $\text{s.t. } X \in \Omega_{\mathcal{I}}$.
In this way, we determine the
nearest matrix in the convex set $\Omega_{\mathcal{I}}$ 
to the zero matrix (that is out of the set), namely the projection of the zero matrix on the set. 
The quantity $\Vert A-B\Vert_F$ is the Frobenius distance between the matrices $A$ and $B$
induced by the 
Hilbert-Schmidt inner product $\tr{A^\dag B}$, from this the term \enquote{Euclidean}.
Since the set is convex, the Euclidean projection is unique.
This unicity, in turn, provides less ambiguity in the construction of an optimal unitary $\tilde{U}_\star$. 
Minimum purity tends to favor high-rank states
$\tilde{\rho}$ with a more uniform spectrum, and, if the rank is maximum and the state and Hamiltonian eigenvalues are nondegenerate, the unitary is uniquely identified. 
Furthermore, the state with minimum purity typically yields smaller ergotropy values among feasible states. Therefore, its ergotropy is not expected to be much greater than the certified ergotropy lower bound obtained in step (ii). This often leads to tighter bounds.

\begin{figure}[h]
    \centering
    \includegraphics[width=0.29\linewidth]{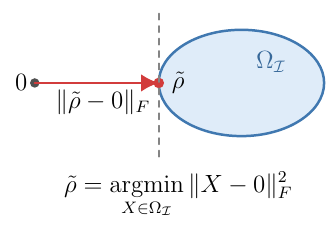}
    \caption{Schematic illustration of the purity minimization viewed as an Euclidean projection of the zero matrix onto the convex set $\Omega_{\mathcal I}$. The projected point $\tilde{\rho}$ is the feasible state minimizing the Frobenius distance from zero, i.e.,
    $\tilde{\rho}=\operatorname*{argmin}_{X\in\Omega_{\mathcal I}}\|X-0\|_F^2$.}
    \label{fig:proj}
\end{figure}

To quantitatively illustrate these considerations,
we compare new results obtained by setting $\ell(\rho)=0$ in step (i) with the results presented in the main text, where we instead use the purity as objective function. 
Such comparison, for the negative-temperature state already considered in Fig.~\ref{fig:no_stat}(b) of the main text, is shown in Fig.~\ref{fig:purity_vs_zero}. 
For a fair benchmark, both curves are obtained using the same randomly selected Pauli strings for each value of $K$, using for simplicity numerically simulated data in the absence of shot noise.
Comparing the two curves, we observe that the lower bound obtained for
$\ell(\rho)=0$
as objective function in step (i)
is slightly less informative than the one obtained considering the purity.

\begin{figure}[h]
\centering
\includegraphics[width=0.5\linewidth]{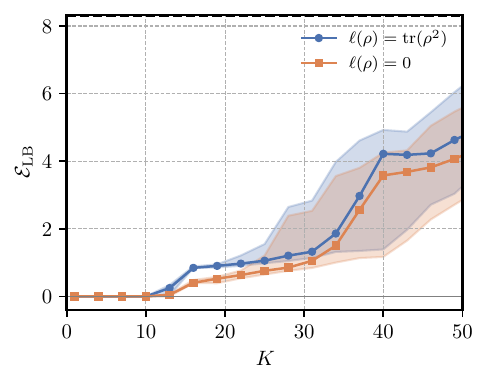}
\caption{Comparison between the ergotropy lower bounds obtained (orange) using $\ell(\rho)=0$ in step (i) and (blue) using the purity $\ell(\rho)=\tr{\rho^2}$. 
We consider here the same settings as in Fig.~\ref{fig:no_stat}(b) of the main text.
We take the negative-temperature state of Eq.~\eqref{negative-gibbs} of the main text ($\beta=-1$) and MFI Hamiltonian ($B=0.5$, $G=0.5$, $\Delta=1$). 
For each value of $K$, each realization is generated according to the hierarchical sampling described in the main text.
The plotted curves show the median over 20 realizations, with shaded regions indicating the interquartile range.
Both curves are computed using the same randomly selected Pauli strings for each value of $K$, to guarantee a fair comparison.
The dashed horizontal line marks the true ergotropy. We considered a system of 5 qubits and numerically simulated measurement data in the absence of shot noise.
}
    \label{fig:purity_vs_zero}
\end{figure}
Interestingly, the choice of minimizing the purity also admits an analytical motivation based on majorization, as discussed in Appendix C of the main text.

\section*{S2. Data generation via binomial statistics}

Concretely, in our setting, an observable $\mathcal{O}_i$ is often taken to be a multiqubit Pauli string $P_i$ (see Figures in the main text), for example $\sigma_x \otimes \cdots \otimes \sigma_z \otimes \mathbb{1}$. 
In particular, for Fig.~\ref{fig:real_stat}(a) of the main text,
simulated measurement data are generated assuming projective measurements of the Pauli strings $\{P_i\}_{i \in \mathcal{I}}$, whose outcomes are $\pm1$.
For a given true state $\rho$, the ideal expectation value is $\langle {P_i} \rangle = {\rm tr}(\rho {P_i})$. Each measurement shot is modeled as an independent Bernoulli trial with probabilities
\begin{align}
p_{i}{}_{+} = \frac{1+\langle {P_i} \rangle}{2}, \qquad
p_{i}{}_{-} = \frac{1-\langle {P_i} \rangle}{2}.
\end{align}
For a total of $N_{i}{}$ measurement shots, the number of $+1$ outcomes is drawn from a binomial distribution $\mathrm{Bin}(N_{i}{},p_{i}{}_{+})$. The empirical expectation value is then estimated as the arithmetic mean
\begin{align}
o^{\rm (est)}
=
\frac{N_{i}{}_{+}-N_{i}{}_{-}}{N}
=
\frac{2N_{i}{}_{+}}{N_{i}{}}-1 ,
\end{align}
where $N_{i}{}_{+}$ and $N_{i}{}_{-}=N_{i}{}-N_{i}{}_{+}$ denote the number of $+1$ and $-1$ outcomes, respectively.

\section*{S3. Pauli strings for the experimental GHZ state}
\label{supp_mat:GHZ_correlators}

We provide here more details on the
experimental ergotropy certification appearing in Fig.~\ref{fig:real_stat}(b) of the main text.
Using the \texttt{ibm\_perth} quantum processor, the four-qubit GHZ state was experimentally prepared through the quantum circuit shown in Fig.~\ref{fig:ghz_circuit}. Starting from the initial state $|0000\rangle$, one first applies a Hadamard gate to qubit $1$, creating the superposition $(|0\rangle+|1\rangle)/\sqrt{2}$. Entanglement is then generated through a sequence of three CNOT gates acting on the pairs $(1,2)$, $(2,3)$, and $(3,4)$. 
The expectation values of the required Pauli observables are obtained by applying suitable local basis-rotation gates $U_{j_i}$ before the final computational-basis measurements, where $U_{j_i}$ is the identity $\mathbb{1}$, the Hadamard gate $U_{\rm Had}$, and $U_{\rm S}{}^\dagger U_{\rm Had}$ ($U_{\rm S}$ being the S gate) for measurements in the $Z$, $X$, and $Y$ basis, respectively. 

\begin{figure}[h]
    \centering
    \includegraphics[width=.65\linewidth]{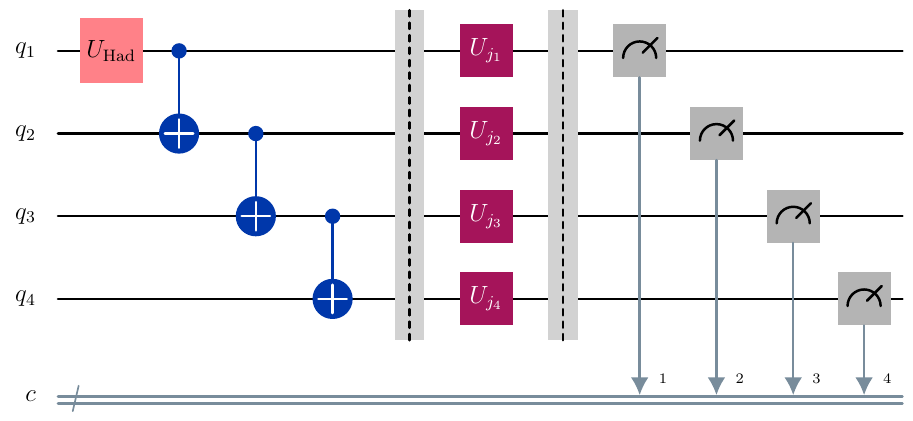}
    \caption{Gate-based quantum circuit used to prepare the four-qubit GHZ state on the IBM quantum processor. After the GHZ state preparation, local single-qubit gates $U_{j_i}$ rotate the measurement basis according to the Pauli observable being measured, followed by computational-basis measurements.}
    \label{fig:ghz_circuit}
\end{figure}
In this way,
we produced an experimental dataset, part of which was previously analyzed in Ref.~\cite{zambrano2024certification} in the context of a different certification problem. The dataset considered here concerns the 60 four-body Pauli observables
$$
\begin{aligned}
\{ &
\text{YYYY}, \text{XXXX}, \text{ZXXY}, \text{YZYX}, \text{XZYZ}, \text{ZYZX}, \text{ZYYZ}, \text{YYYY},\\
& \text{ZXXY}, \text{XXXX}, \text{YZYX}, \text{ZYZX}, \text{XZYZ}, \text{ZYZY}, \text{ZXXY}, \text{YXZZ},\\
& \text{XYXX}, \text{ZZZZ}, \text{XZXZ}, \text{ZYZX}, \text{XXYY}, \text{YZXY}, \text{ZZYX}, \text{XYXX},\\
& \text{YZZZ}, \text{ZZYY}, \text{ZZZY}, \text{YZXY}, \text{XZZY}, \text{ZYYX}, \text{YXXX}, \text{ZZYY},\\
& \text{XYZY}, \text{XXZZ}, \text{XZXY}, \text{ZYYX}, \text{YXXZ}, \text{YYXZ}, \text{XZXX}, \text{XXZZ},\\
& \text{XXXY}, \text{YZZY}, \text{ZYYY}, \text{YYXZ}, \text{YYXX}, \text{YZZY}, \text{YZXZ}, \text{YXYX},\\
& \text{YXYX}, \text{YZYZ}, \text{ZZXZ}, \text{ZXXZ}, \text{XYXZ}, \text{ZXXZ}, \text{ZXZZ}, \text{YZZX},\\
& \text{XZZX}, \text{YZZX}, \text{XYYZ}, \text{XYZZ} 
\}.
\end{aligned}
$$
Considering $K$ observables (abscissa of Fig.~\ref{fig:real_stat}(b) of the main text) here means considering the first $K$ entries of the above list, and estimating their expectation values. For a given $K$, this defines the set $\Omega_\mathcal{I}$ of our two-step protocol in the presence of finite statistics [see Eq.~\eqref{feasible-set-conc-ineq} of the main text].

\section*{S4. Proof of Lemma 1}

We report the proof of Lemma~\ref{lemma1} associated to Eq.~\eqref{theorem:IC} of the main text.
\begin{proof}
Let $\omega\in\Omega_{\mathcal I}$ be arbitrary. The completely dephasing map in the energy basis acts as follows:
\[
\mathcal D(\omega)=\sum_{i=1}^d \bra{\epsilon_i}\omega\ket{\epsilon_i}\,\ket{\epsilon_i}\!\bra{\epsilon_i}.
\]
By construction, $\mathcal D(\omega)=\rho^\ast$ for all $\omega\in\Omega_{\mathcal I}$. Since $\mathcal D$ is unital, it is a mixing channel; hence $\omega \succ \rho^\ast$ (majorization)~\cite{Marshall,Nielsen2002}. From Eq.~\eqref{majorization} of the main text, we get $\tr{\omega_{\!p}H}\le \tr{\rho^\ast_{\!p}H}$.

All states in $\Omega_{\mathcal I}$ have the same average energy $\sum_i p_i\epsilon_i$, so the ergotropy is minimized when the passive energy is maximized, i.e., at $\rho^\ast$.
For a Hamiltonian $H$ with $\epsilon_1\le\cdots\le\epsilon_d$, the passive state of $\rho^\ast$ is obtained by sorting its populations in non-increasing order; hence
$\tr{\rho^\ast_{\!p}H}=\sum_{i=1}^d p_i^{\downarrow}\,\epsilon_i,$
which yields
\[
\bar{\mathcal{E}}_{\rm LB}=\min_{\omega\in\Omega_{\mathcal I}} \mathcal E(\omega)
=\mathcal E(\rho^\ast, H)
=\sum_{i=1}^d p_i\,\epsilon_i - \sum_{i=1}^d p_i^{\downarrow}\,\epsilon_i.
\]
The bound is tight because it is attained by $\rho^\ast$. Finally, by definition, the incoherent ergotropy of $X$ equals the ergotropy of its dephased state in the energy basis~\cite{Francica'2020}, 
so $\bar{\mathcal{E}}_{\rm LB}=\mathcal E(\rho^\ast, H)=\mathcal{E}_{IC}(X)$.
\end{proof}

\section*{S5. Qubit case: nonzero coherent ergotropy}
\label{sec:AppendixC}

When the Hamiltonian is diagonal in the computational basis, the mean energy of a qubit is fully determined by the expectation value of Pauli $\sigma^z$. However, measuring $\sigma^z$ alone cannot certify any \emph{coherent} contribution to ergotropy, since the off-diagonal elements in the energy basis remain unconstrained. To certify a coherence-induced contribution, one must also measure an observable that does not commute with $\sigma^z$, e.g., Pauli $\sigma^x$.

Assume that for an unknown qubit state $\rho$ we have measured
\[
x^\ast=\tr{\sigma^x\rho},\qquad z^\ast=\tr{\sigma^z\rho}.
\]
Any state $\omega\in\Omega_{\mathcal I}$ consistent with these constraints can be written as
\begin{equation}
\omega(y)\;=\;
\begin{pmatrix}
\frac{1+z^\ast}{2} & \frac{x^\ast+ i y}{2} \\
\frac{x^\ast- i y}{2} & \frac{1-z^\ast}{2}
\end{pmatrix},
\qquad
|y|\;\le\;\sqrt{1-(x^\ast)^2-(z^\ast)^2},
\end{equation}
where $y=\tr{\sigma^y \omega}$ is the remaining free parameter and is unknown to us. Thus, the feasible set is the line segment along the $Y$-axis inside the Bloch ball determined by fixing $x^\ast$ and $z^\ast$. The endpoints
\[
y=\pm \sqrt{1-(x^\ast)^2-(z^\ast)^2}
\]
are pure, and we denote the corresponding states by $\ket{y_\pm}\!\bra{y_\pm}$.

The midpoint of the segment, obtained at $y=0$, is denoted by $\rho^\ast_{x^\ast,z^\ast}$ and is given by
\begin{equation}\label{eq:rho-star-xz}
\rho^\ast_{x^\ast,z^\ast} \;=\;
\begin{pmatrix}
\frac{1+z^\ast}{2} & \frac{x^\ast}{2} \\
\frac{x^\ast}{2} & \frac{1-z^\ast}{2}
\end{pmatrix}
\;=\;
\frac{1}{2}\ket{y_+}\!\bra{y_+}+\frac{1}{2}\,U\,\ket{y_+}\!\bra{y_+}\,U^\dagger,
\end{equation}
where $U=\exp\!\bigl(-i\frac{\pi}{2}\,\hat n\cdot\vec\sigma\bigr)$ is the $\pi$-rotation about the Bloch axis
$\hat n\propto (x^\ast,0,z^\ast)$. This unitary maps $\ket{y_+}$ to $\ket{y_-}$ while leaving $x^\ast$ and $z^\ast$ invariant.

\begin{lemma}
For a qubit with Hamiltonian diagonal in the computational basis, and constraints $x^\ast=\tr{\sigma^x\rho}$ and $z^\ast=\tr{\sigma^z\rho}$, the minimum ergotropy over the feasible set $\Omega_{\mathcal I}$ is attained at $\rho^\ast_{x^\ast,z^\ast}$, and the lower bound is tight:
\[
\bar{\mathcal{E}}_{\mathrm{LB}} \;=\; \min_{\omega\in\Omega_{\mathcal I}}\mathcal{E}(\omega)
\;=\; \mathcal{E}\!\left(\rho^\ast_{x^\ast,z^\ast},H\right).
\]
Moreover,
\[
\bar{\mathcal{E}}_{\mathrm{LB}} > \mathcal{E}_{IC}(\omega)
\quad \text{if and only if} \quad x^\ast > 0,
\]
where $\omega$ is the true unknown state.
\end{lemma}

\begin{proof}
Averaging over $\{\,\mathbb{1},U\,\}$ defines a unital (mixing) channel that sends any $\omega(y)$ to $\rho^\ast_{x^\ast,z^\ast}$:
\[
\mathcal T(\omega)=\tfrac12\bigl(\omega+U\omega U^\dagger\bigr)=\rho^\ast_{x^\ast,z^\ast}.
\]
Since $\mathcal T$ is mixing, $\omega \succ \rho^\ast_{x^\ast,z^\ast}$ (majorization) for all $\omega\in\Omega_{\mathcal I}$. For qubit Hamiltonians diagonal in the computational basis, majorization implies an ordering of passive energies, hence
\[
\tr{\omega_p H}\le \tr{(\rho^\ast_{x^\ast,z^\ast})_p H}.
\]
All states in $\Omega_{\mathcal I}$ share the same average energy (fixed by $z^\ast$), so minimizing ergotropy over $\Omega_{\mathcal I}$ is equivalent to maximizing the passive energy. Therefore, the minimum is attained at $\rho^\ast_{x^\ast,z^\ast}$ and
\[
\bar{\mathcal{E}}_{\mathrm{LB}}
=
\min_{\omega\in\Omega_{\mathcal I}}\mathcal{E}(\omega)
=
\mathcal{E}\!\left(\rho^\ast_{x^\ast,z^\ast},H\right),
\]
which proves tightness.

For the coherent part, note that $\rho^\ast_{z^\ast}$ is the dephased state. As shown in Eq.~\eqref{theorem:IC} of the main text, we already know that the certified ergotropy satisfies $\mathcal{E}(\rho^\ast_{z^\ast},H)=\mathcal{E}_{IC}(\omega)$, where $\omega$ is taken to be the true state. Hence, the difference
\[
\mathcal{E}\!\left(\rho^{\ast}_{x^\ast,z^\ast},H\right)
-
\mathcal{E}\!\left(\rho^{\ast}_{z^\ast},H\right)
\]
quantifies the certified gain beyond the incoherent contribution, and it provides a lower bound on the coherently extractable ergotropy $\Delta_C(\omega)$.

The bound can be evaluated explicitly as
\begin{align}
\Delta_C (\omega)
&\geq \mathcal{E}\!\left(\rho^{\ast}_{x^\ast,z^\ast},H\right)
 - \mathcal{E}_{IC}(\omega) \\
&= \mathcal{E}
\begin{pmatrix}
\tfrac{1+z^\ast}{2} & \tfrac{x^\ast}{2} \\
\tfrac{x^\ast}{2} & \tfrac{1-z^\ast}{2}
\end{pmatrix}
-
\mathcal{E}
\begin{pmatrix}
\tfrac{1+z^\ast}{2} & 0 \\
0 & \tfrac{1-z^\ast}{2}
\end{pmatrix} \\
&= \sum_i (p_i - \lambda_i)\,\epsilon_i > 0
\qquad
\text{if and only if $x^\ast > 0$ in the unknown true state $\omega$}.
\end{align}

Here $\{\epsilon_i\}$ are the system's energy levels (assume $\epsilon_0 \le \epsilon_1$).
The diagonal state $\rho^{\ast}_{z^\ast}$ has populations
\[
p_0=\frac{1+z^\ast}{2},\qquad
p_1=\frac{1-z^\ast}{2},
\]
while the eigenvalues of $\rho^{\ast}_{x^\ast,z^\ast}$ are
\[
\lambda_{0,1}
=
\frac{1\pm \sqrt{(x^\ast)^2+(z^\ast)^2}}{2}.
\]
\end{proof}

\end{document}